\documentclass[journal]{IEEEtran}
\usepackage{amsmath,amsthm,amssymb,color,graphicx,caption,subcaption,comment,tikz,mathtools}

\usepackage{pgfplots}
\pgfplotsset{compat=newest}
\usepackage{verbatim}
\usetikzlibrary{decorations.markings}
	
\usepackage{mathtools, cuted}

\usepgfplotslibrary{units}
\usetikzlibrary{spy,backgrounds}
\usepackage{pgfplotstable}
\usepackage{cite}

\newtheorem{theorem}{Theorem}

\newtheorem{definition}{Definition}
\newtheorem{remark}{Remark}

\ifCLASSINFOpdf
\else
\fi
\renewcommand{\S}{\mathcal{S}}
\newcommand{\T}{\mathcal{T}}
\newcommand{\M}{\mathcal{M}}

\newcounter{tempEquationCounter} 
\newcounter{thisEquationNumber}

\begin{document}
%
\title{Capacity Bounds on the Downlink of Symmetric, Multi-Relay, Single Receiver C-RAN Networks}
\title{Capacity Bounds on the Downlink of Symmetric, Multi-Relay, Single Receiver C-RAN Networks}
\author{Shirin~Saeedi Bidokhti,
        Gerhard~Kramer
      and Shlomo~Shamai (Shitz) 
\thanks{S. Saeedi Bidokhti is with the Department of Electrical Engineering at Stanford University, USA. G. Kramer is with the Department for Electrical and Computer Engineering, Technical University of Munich, Germany. S.~Shamai  is with the Department of Electrical Engineering, Technion, Israel (saeedi@stanford.edu, gerhard.kramer@tum.de, sshlomo@ee.technion.ac.il).}
\thanks{{The work of S. Saeedi Bidokhti was supported by the Swiss National Science Foundation fellowship no. 158487. The work of G. Kramer was supported by the German Federal Ministry of Education and Research in the framework of the Alexander von Humboldt-Professorship. The work of S. Shamai was supported by the
European Union's Horizon 2020 Research and Innovation Program,
grant agreement no. 694630.}}}
\maketitle

\begin{abstract}
The downlink of symmetric Cloud Radio Access Networks (C-RANs) with multiple relays and a single receiver is studied. Lower and upper bounds are derived on the capacity. The lower bound is achieved by Marton's coding which facilitates dependence  among  the multiple-access channel inputs. The upper bound uses Ozarow's technique to augment the system with an auxiliary random variable. The bounds are studied over scalar Gaussian C-RANs and are shown to meet and characterize the capacity for interesting regimes of operation.
\end{abstract}


%
\IEEEpeerreviewmaketitle

%
%
%
%



\section{Introduction}
Cloud Radio Access Networks (C-RANs) are expected to be a part of  future mobile network architectures.  In  C-RANs,  information processing is done in a cloud-based central unit that is connected to remote radio heads (or relays) by rate-limited fronthaul links. C-RANs improve energy and bandwidth efficiency and reduce complexity of relays by facilitating centralized information processing and cooperative communication. We refer to\cite{ParkSimeoneSahinShamai14,Yu16, PengWangLauPoor15} and the references therein for an overview of the challenges and coding techniques for C-RANs. 
 Several coding schemes have been proposed in recent years for the downlink of C-RANs including message sharing \cite{DaiYu14}, backhaul compression \cite{Parkatal13}, hybrid schemes \cite{PatilYu14}, and generalized data sharing using Marton's coding \cite{LiuKang14, WangWiggerZaidi16}. 
While none of these schemes are known to be optimal,  \cite{YangLiuKangShamai16} has proved an upper bound on the sum-rate  of $2$-relay C-RANs with two users and numerically compared the performance of the aforementioned schemes with the upper bound.

We consider the downlink of a C-RAN with multiple relays and a single user. This network may be modeled by an  $M$-relay diamond network where the broadcast component is modeled by rate-limited  links and the multiaccess component is modeled by a memoryless multiple access channel (MAC), see Fig. \ref{two-user-mac}. The capacity of this class of  networks is not known in general, but lower and upper bounds were derived in \cite{TraskovKramer07,KangLiuChong15,SaeediKramer16} for $2$-relay networks. Moreover, the capacity was found for binary adder MACs \cite{SaeediKramer16}, and for certain regimes of operation in Gaussian MACs \cite{KangLiuChong15,SaeediKramer16}. 
In this work, we derive  lower and upper bounds for symmetric C-RANs with multiple relays and find the capacity in interesting regimes of operation for symmetric Gaussian C-RANs. 

The rest of the paper is organized as follows. Section \ref{prel} introduces  notation and the problem setup. In Section~\ref{sec:lowerbound}, we propose a coding strategy based on Marton's coding  and discuss simplifications for symmetric networks. In Section~\ref{secupp}, we generalize the bounding technique in \cite{KangLiuChong15,SaeediKramer16}. The case of Gaussian C-RANs is studied in Section \ref{sec:GaussianCRAN}, where we compute  lower and upper bounds and show that they meet in certain  regimes of operation characterized in terms of power,  number of users, and broadcast link capacities.
\section{Preliminaries and Problem Setup}
\label{prel}
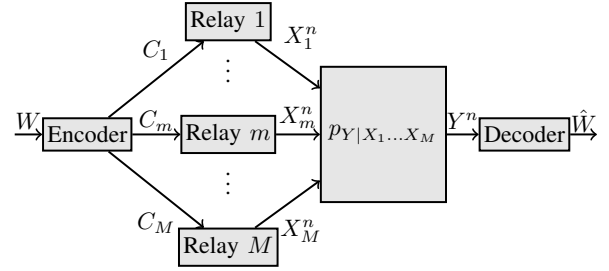
\begin{figure}[t!]
\centering
\begin{tikzpicture}[scale=.75]
\definecolor{light-gray}{gray}{0.9}
\tikzstyle{every node}=[draw,,line width=.75pt,shape=circle,font=\normalsize,scale=.9];

\path (-1.5,1.5) node[draw=none] (s0) {};
\path (0,1.5) node[shape=rectangle,fill=light-gray] (s) {$\hspace{-.05cm}\text{Encoder}\hspace{-.05cm}$};

\path (2.5,3.5) node[shape=rectangle,fill=light-gray] (r1) {$\hspace{-.05cm}\text{Relay  $1$}\hspace{-.05cm}$};
\path (2.5,1.5) node[shape=rectangle,fill=light-gray] (r2) {$\hspace{-.05cm}\text{Relay  $m$}\hspace{-.05cm}$};
\path (2.5,-.5) node[shape=rectangle,fill=light-gray] (rm) {$\hspace{-.05cm}\text{Relay  $M$}\hspace{-.05cm}$};
\path (2.5,.75) node[draw=none] (vd) {$\vdots$};
\path (2.5,2.75) node[draw=none] (vd) {$\vdots$};

\path (5.25,1.5) node[rectangle,minimum height=2cm,fill=light-gray] (mac) {$p_{Y|X_1\ldots X_M}$};

\path (7.75,1.5) node[shape=rectangle,fill=light-gray] (d) {$\hspace{-.05cm}\text{Decoder}\hspace{-.05cm}$};
\path (9.25,1.5) node[draw=none] (ds) {};

\draw[->,line width=.75pt] (s0) --node[draw=none,yshift=.2cm]{$W$} (s);  
\draw[->,line width=.75pt] (d) --node[draw=none,yshift=.2cm]{$\hat{W}$} (ds);  
\draw[->,line width=.75pt] (s) --node[above,draw=none]{$C_1$} (r1);
\draw[->,line width=.75pt] (s) --node[yshift=.2cm,draw=none]{$C_m$} (r2);
\draw[->,line width=.75pt] (s) --node[below,,draw=none]{$C_M$} (rm);
\draw[->,line width=.75pt] (r1) --node[draw=none,xshift=.2cm,yshift=.4cm]{$X^n_1$} (mac);
\draw[->,line width=.75pt] (r2) --node[draw=none,yshift=.25cm]{$X^n_m$} (mac);
\draw[->,line width=.75pt] (rm) --node[draw=none,xshift=.2cm,yshift=-.4cm]{$X^n_M$} (mac);
\draw[->,line width=.75pt] (mac) --node[draw=none,yshift=.2cm]{$Y^n$} (d);

\end{tikzpicture}
\caption{A C-RAN downlink.}
\label{two-user-mac}
\vspace{-.2cm}
\end{figure}
\subsection{Notation}
Random variables are denoted by uppercase letters, e.g. $X$,  their realizations are denoted by lowercase letters, e.g. $x$, and their corresponding probabilities are denoted by $p_X(x)$ or $p(x)$. The probability mass function (pmf) describing $X$ is denoted by $p_X$. $\mathcal{T}^n_\epsilon(X)$ denotes the set of sequences that are $\epsilon$-typical with respect to $P_X$ \cite[Page 25]{ElGamalKim}. When $P_X$ is clear from the context we write $\mathcal{T}^n_\epsilon$.
The entropy of $X$ is denoted by $H(X)$,  the conditional entropy of $X$ given $Y$ is denoted by $H(X|Y)$ and the mutual information between $X$ and $Y$ is denoted by $I(X;Y)$. Similarly, differential entropies  and conditional differential entropies are denoted by $h(X)$ and $h(X|Y)$. 

Matrices are denoted by bold letters, e.g. $\mathbf{K}$. We denote the entry of matrix $\mathbf{K}$ in row $i$ and column $j$ by $\mathbf{K}_{ij}$. Sets are denoted by script letters, e.g., $\mathcal{S}$. The cartesian product of $\mathcal{S}_1$ and $\mathcal{S}_2$ is denoted by $\mathcal{S}_1\times\mathcal{S}_2$, and the $n$-fold Cartesian product of $\mathcal{S}$ is denoted by $\mathcal{S}^n$. The cardinality of $\mathcal{S}$  is denoted by $|\mathcal{S}|$.  

Given the set $\mathcal{S}=\{s_1,\ldots,s_{|\mathcal{S}|}\}$, $X_{\mathcal{S}}$ denotes the tuple $(X_{s_1},\ldots,X_{s_{|\mathcal{S}|}})$.  The random string $X_1,\ldots,X_n$ is denoted by $X^n$.  $I(X_\S)$ is defined as follows (see \cite[Eqn (74)]{VenkataramaniKramerGoyal03}):
 \begin{align}
\label{eq:multi-inf}
I(X_\S)=\sum_{m\in\S}H(X_m)-H(X_\S).
\end{align}
For example, when $\mathcal{S}=\{s_1,s_2\}$, \eqref{eq:multi-inf} becomes the mutual information $I(X_{s_1};X_{s_2})$. The conditional version of~\eqref{eq:multi-inf}, $I(X_\S|U)$, is defined similarly by conditioning all terms in~\eqref{eq:multi-inf} on $U$. Note that  $I(X_\S|U)$ is non-negative.


\subsection{Model}
Consider the C-RAN  in Fig. \ref{two-user-mac}, where a source communicates a message $W$ with $nR$ bits to a sink with the help of $M$ relays.  
Let  $\M=\{1,\ldots,M\}$ be the set of relays.

The source encodes $W$ into descriptions $\mathbb{V}_1,\ldots,\mathbb V_M$ that are provided to relays $1,\ldots,M$, respectively. We focus mainly on symmetric networks where $\mathbb V_m$ satisfies
\begin{align}
H(\mathbb V_m)\leq nC,\qquad m=1,\ldots,M.
\end{align}
Each relay $m$, $m=1,\ldots,M$,  maps its  description $\mathbb V_m$ into a sequence $X_m^n$ which is sent over a  multiple access channel. The MAC is characterized by the  input alphabets $\mathcal{X}_1, \ldots, \mathcal{X}_M$, the output alphabet $\mathcal{Y}$, and the transitional probabilities $p(y|x_1,\ldots,x_M)$ for all $(x_1,\ldots,x_M,y)\in\mathcal{X}_1\times\ldots\times\mathcal{X}_M\times\mathcal{Y}$. From the received sequence $Y^n$, the sink decodes an estimate $\hat{W}$ of $W$. 

A coding scheme consists of an encoder, $M$ relay mappings, and a decoder, and is said to achieve
the rate $R$ if, by choosing $n$ sufficiently large, we can make the error probability $\Pr(\hat{W}\neq W)$ as small as desired. We are interested in characterizing the largest achievable rate $R$. We refer to the maximum  rate as the capacity $C^{(M)}$ of the network.

In this work, we focus on symmetric networks:
\begin{definition}
\label{def:symnet}
The network in Fig. \ref{two-user-mac} is symmetric if we have
\begin{align}
&C_1=\ldots=C_M=:C\\
&\mathcal{X}_1=\ldots=\mathcal{X}_M=:\mathcal{X}
\end{align}
and 
\begin{align}
p_{Y|X_1\ldots X_M}(y|x_1,\ldots,x_M)\!=\!p_{Y|X_1\ldots X_M}(y|x^\prime_1,\ldots,x^\prime_M)\label{chansym}
\end{align}
for all $y\in\mathcal{Y}$, $(x_1,\ldots,x_M)\in \mathcal{X}^M$ and any of its permutations $(x^\prime_1,\ldots,x^\prime_M)$.
\end{definition}

When the MAC is Gaussian, the input and output alphabets are the set of real numbers and the output is given by
\begin{align}
&Y=\sum_{m=1}^MX_M+Z\label{eq:GMACdef1}
\end{align}
where $Z$ is Gaussian noise with zero mean and unit variance.  We consider average block power constraints $P_1,\ldots,P_M$:
\begin{align}
&\frac{1}{n}\sum_{i=1}^n\mathbb{E}(X_{m,i}^2)\leq P_m,\qquad m=1,\ldots,M.\label{eq:GMACdef2}
\end{align}
The Gaussian C-RAN is symmetric if $C_m=C$ and $P_m=P$ for all $m=1,\ldots,M$.

\section{A Lower Bound}
\label{sec:lowerbound}

We outline an achievable scheme based on Marton's coding. We remark that this scheme can be improved for certain regimes of $C$ by using superposition coding (e.g., see \cite[Theorem 2]{SaeediKramer14} and \cite[Theorem 2]{KangLiuChong15}). 

Fix the pmf $p(x_1,\ldots,x_M)$, $\epsilon>0$, and the auxiliary rates $R_m,R^\prime_m$, $m=1,\ldots,M$, such that
\begin{align}
& R_m,R^\prime_m\geq 0\label{en:nonneg}\\
&R_m+R^\prime_m\leq C_m.\label{en:set1}
\end{align}
\subsubsection{Codebook construction}
Set 
\begin{align}
&R=\sum_{m=1}^MR_m.\label{en:set2}
\end{align}
For every $m=1,\ldots,M$, generate $2^{n(R_m+R^\prime_m)}$ sequences $x_m^n(w_m,w^\prime_m)$, $w_M=1,\ldots,2^{nR_m}$, $w^\prime_M=1,\ldots,2^{nR^\prime_m}$,  in an i.i.d manner according to $\prod_{\ell}P_{X_\ell}(x_{m,\ell})$, independently across $m=1,\dots,M$. For each bin index $(w_{1},\ldots,w_{M})$, pick a jointly typical sequence tuple 
\begin{align}
\label{en:jt1}
(x_1^n(w_1,w^{\prime\star}_1),\ldots,x_M^n(w_M,w^{\prime\star}_M))\in \mathcal{T}_\epsilon^n.
\end{align}

\subsubsection{Encoding}
Represent  message $w$ as a tuple $(w_1,\ldots,w_M)$, and send $(w_m,w^{\prime\star}_m)$ to relay $m$, $m=1,\ldots,M$. 
\subsubsection{Relay mapping at relay $m$, $m=1,\ldots,M$}
Relay $m$ sends $X_m^n(w_m,w^{\prime\star}_m)$ over the MAC.

\subsubsection{Decoding}
Upon receiving $y^n$, the receiver looks for indices $\hat{w}_1,\ldots,\hat{w}_M$ for which the following joint typicality test holds for some  $\hat{w}^\prime_1,\ldots,\hat{w}^\prime_M$:
\begin{align}
\label{de:jt2}
(x_1^n(\hat{w}_1,\hat{w}^{\prime}_1),\ldots,x_M^n(\hat{w}_M,\hat{w}^{\prime}_M),y^n)\in \mathcal{T}_\epsilon^n.
\end{align}

We show in Appendix \ref{chert} 
that the above scheme has a vanishing error probability as $n\to \infty$ if in addition to \eqref{en:nonneg}--\eqref{en:set2} we have
\begin{align}
\sum_{m\in\S}\!R^\prime_m\!\geq& I(X_\S),\qquad \forall\S\subseteq\M\label{en:en}\\
\sum_{m\in\S}\! R_m\!+\!R^\prime_m\!\leq& I(X_{{\S}};Y|X_{\bar{\S}})\!+\!I(X_\M)\!-\!I(X_{\bar{\S}}),\quad \forall\S\subseteq\M\label{en:de}.
\end{align}

One can use  Fourier-Motzkin elimination  to eliminate $R_m,R^\prime_m$, $m=1,\ldots,M$, from \eqref{en:nonneg}--\eqref{en:set2}, \eqref{en:en}, \eqref{en:de}, and characterize the set of achievable rates $R$.
For  symmetric networks (see Definition \ref{def:symnet}), we  bypass the above step and proceed by  choosing  $p_{X_\M}$ to be ``symmetric".
We say $p_{X_\M}$ is symmetric if 
\begin{align}
\mathcal{X}_1=\ldots=\mathcal{X}_M=\mathcal{X}
\end{align}
and for all subsets $\S, \S^\prime\subseteq\M$ with $|\S|=|\S^\prime|$ we have
\begin{align}
p_{X_\S}(x_1,\ldots,x_{|\S|})=p_{X_{\S^\prime}}&(x_1,\ldots,x_{|\S|}),\nonumber\\&\quad \forall (x_1,\ldots,x_{|\S|})\in \mathcal{X}^{|\S|}.\label{symm1}
\end{align}
We simplify the problem defined by \eqref{en:nonneg}--\eqref{en:set2}, \eqref{en:en}, \eqref{en:de} for symmetric distributions and prove the following result in Appendix~\ref{ap-simplify}. 
\begin{theorem}
\label{thm:LB}
For symmetric C-RAN downlinks, the rate $R$ is achievable if
\begin{align}
&R\leq MC-I(X_\M)\label{reg:3user1}\\
&R\leq I(X_{\M};Y)\label{reg:3user2}
\end{align}
for some symmetric distribution $p_{X_\M}$.
\end{theorem}

%
%

\section{An upper bound}
\begin{figure*}
\begin{align}
C^{(M)}\leq\max_{p(\mathbf{x})}\min_{\substack{p(u|\mathbf{x}y)\\=p(u|y)}}\max_{\substack{p(q|\mathbf{x},u,y)\\=p(q|\mathbf{x})}}\left\{\begin{array}{l} MC -(M\!-\!1) H(U|Q)+\sum_{m=1}^M H(U|X_mQ)-H(U|X_\M )\\
\min_{\S\subseteq\M}|\S|C+I(X_{\bar{\S}};Y|X_\S Q)\end{array}\right\}.\label{thm:upperbound}
\end{align}
\hrule
\vspace{-.2cm}
\end{figure*}

\label{secupp}
Our upper bound is motivated by \cite{SaeediKramer16,VenkataramaniKramerGoyal03,Ozarow80,KangLiuChong15}.
\begin{theorem}
\label{thm:upperbound-statement}
The capacity $C^{(M)}$ is upper bounded as shown in \eqref{thm:upperbound} on the top of this page, where $Q-X_\M-Y-U$ forms a Markov chain. Moreover, the alphabet size of $Q$ may be chosen to satisfy {$|\mathcal{Q}|\leq \prod_{i=1}^M|\mathcal{X}_i|+2^M-1 $}.
\end{theorem}


%

\begin{remark}
For $M\!=\!2$, Theorem \ref{thm:upperbound-statement} reduces to \cite[Theorem~3]{SaeediKramer16}.
\end{remark}

\begin{proof}[Proof Outline]
We start with the following $n$-letter upper bound (see Appendix \ref{ap-thm:upperbound-statement}): 
\begin{align}
nR&\leq nMC\!-\!I(X_\M^n)\label{eq:ubound1}\\
nR&\leq n|\S|C\!+\!I(X_\M^n;Y^n|X_\S^n),\quad\forall\S\!\subseteq\! \M.\label{eq:ubound2}
\end{align}

For any sequence $U^n$, we have
\begin{align}
I(X_\M^n)\geq& I(X_\M^n)- I(X_\M^n|U^n)\nonumber\\
=&\left[\sum_{m=1}^M I(X_{m}^n;U^n)\right]\!-I(X_\M^n;U^n)\nonumber\\
=& (M\!-\!1) H(U^n)\!-\!\!\left[\sum_{m=1}^M \!H(U^n|X_m^n)\!\right]\!+\!H(U^n|X_\M^n).\label{eq:Xm}
\end{align}
By substituting \eqref{eq:Xm} into \eqref{eq:ubound1}, we obtain
\begin{align}
nR\leq& nMC\!-\!(M\!-\!1) H(U^n)\nonumber\\&+\left[\sum_{m=1}^M\! H(U^n|X_m^n)\!\right]\!\!-\!H(U^n|X_\M^n).\label{rewriteU}
\end{align}

We now choose $U_i$, $i=1,\ldots,n$, to be the output of a memoryless channel $p_{U|Y}(u_i|y_i)$ with  input $Y_i$. The auxiliary channel $p_{U|Y}(.|.)$ will be optimized later. With this choice, we single letterize \eqref{eq:ubound1} and \eqref{eq:ubound2} and obtain
\begin{align}
R\leq& MC \!-\!(M\!-\!1) H(U|Q)\!+\!\sum_{m=1}^M H(U|X_mQ)\!-\!H(U|X_\M )\label{upp1}\\
R\leq& |\S|C+I(X_{\bar{\S}};Y|X_\S Q),\qquad \forall \S\subseteq\M.\label{upp2}
\end{align}
Details of the proof are presented in Appendix \ref{ap-thm:upperbound-statement}.
\end{proof}

\section{The Symmetric Gaussian C-RAN}
\label{sec:GaussianCRAN}
First, we specialize Theorem \ref{thm:LB} to the symmetric Gaussian C-RAN  defined in \eqref{eq:GMACdef1}--\eqref{eq:GMACdef2} where $P_m=P$ for all $m=1,\ldots,M$. Choose $(X_1,\ldots,X_M)$ to be jointly Gaussian with the covariance matrix
\begin{align}
\mathbf{K}_{M}(\rho)=\left[\begin{array}{cccc}P & \rho P &\ldots &\rho P\\ \rho P&P&\ldots& \rho P\\\vdots&&\ddots&\vdots\\ \rho P&\ldots&\rho P&P\end{array}\right].\label{cov-mat}
\end{align}

 \begin{remark}
Choosing $(X_1,\ldots,X_M)$ to be jointly Gaussian (and/or symmetric) is not necessarily optimal for \eqref{en:en}--\eqref{en:de}, but it  gives a lower bound on the capacity.
 \end{remark}
\begin{theorem}
\label{cor:lb-Gaussian}
The rate $R$ is achievable if it satisfies the following constraints  for some non-negative parameter $\rho$, $0\leq \rho\leq 1$:
\begin{align}
&R\leq MC-\frac{1}{2}\log\left(\frac{P^M}{\det(\mathbf{K}_{M}(\rho))}\right)\label{const1}\\
&R\leq \frac{1}{2}\log (1+PM(1+(M-1)\rho)).\label{const3}
\end{align}
\end{theorem}
\begin{remark}
One can recursively calculate $\det(\mathbf{K}_{M}(\rho))$: 
\begin{align}
\det(\mathbf{K}_{M}(\rho))=&P^M\left(1-\rho^2\sum_{i=1}^{M-1}i(1-\rho)^{i-1}\right)\nonumber\\
=&P^M(1\!-\!\rho)^{M-1}\left(1\!+\!(M\!-\!1)\rho\right).
\end{align}
\end{remark}

Let $R^{(\ell)}$ be the maximum achievable rate given by \eqref{const1}--\eqref{const3}. The RHS of \eqref{const1} is non-increasing in $\rho$ and the RHS of \eqref{const3} is increasing in $\rho$. Therefore, we have the following two cases for the optimizing solution $\rho^{(\ell)}$:
\begin{itemize}
\item If $MC\leq \frac{1}{2}\log\left(1+PM\right)$ then 
\begin{align}\rho^{(\ell)}=0\quad \text{and} \quad R^{(\ell)}=MC.\label{eq:rl0}
\end{align}
\item Otherwise, $\rho^{(\ell)}$ is the unique solution of $\rho$ in
\begin{align}
&\frac{1}{2}\log (1+PM(1+(M-1)\rho))\nonumber\\&\qquad=MC\!-\!\frac{1}{2}\log\!\left(\!\frac{1}{(1\!-\!\rho)^{M\!-\!1}\left(1+(M\!-\!1)\rho\right)}\!\right)\!\label{def:rl}
\end{align}
and we have
\begin{align}
R^{(\ell)}=\frac{1}{2}\log (1+PM(1+(M-1)\rho^{(\ell)})).\label{def:Rl}
\end{align}
\end{itemize}

We next specialize Theorem \ref{thm:upperbound-statement} to symmetric Gaussian~{C-RANs}. 
\begin{theorem}
\label{cor-up-Gaussian}
The rate $R$ is achievable only if there exists $\rho$, $0\leq \rho\leq 1$, such that the following inequalities hold for all $N\geq 0$:
\begin{align}
R\leq&MC -(M\!-\!1) \frac{1}{2}\log\left(2^{2R}+  N\right)-\frac{1}{2}\log(1+N)\nonumber\\\quad&\!+\! \frac{M}{2}\!\log\left(1\!+\!N\!+\!P\!\left((M\!-\!1)(1\!-\!\rho)(1\!+\!(M\!-\!1)\rho)\hspace{-.05cm}\right)\right)\!\!\label{bound:R}\\
R\leq& \frac{1}{2}\log\left( 1+PM(1+(M-1)\rho\right).\label{bound:cutR}
\end{align}
\end{theorem}

\begin{proof}[Proof of Theorem \ref{cor-up-Gaussian}]
Set
\begin{align}
U_i=Y_i+Z_{N,i},\qquad i=1,\ldots,n\label{GaussianU}
\end{align}
where $\{Z_{N,i}\}_{i=1}^n$ are identically distributed according to the normal distribution $\mathcal{N}(0,N)$ and are independent from each other and  $X_1^n,\ldots,X_M^n$. The variance $N$ is to be optimized. 


In order to find a computable upper bound in \eqref{thm:upperbound}, we need to lower bound $h(U|Q)$.
Recall that $U$ is a noisy version of $Y$. We thus use the conditional  entropy power inequality \cite[Theorem 17.7.3]{CoverThomas}:
\begin{align}
h(U|Q)\geq& \frac{1}{2}\log\left(2^{2h({Y}|Q)}+2\pi e N\right)\nonumber\\
\geq&\frac{1}{2}\log\left(2\pi e \left(2^{2R}+  N\right)\right).\label{eq:epi}
\end{align}
 Substituting \eqref{eq:epi} into the first constraint of \eqref{thm:upperbound}, we obtain:
\begin{align}
R\leq& MC -(M\!-\!1) \frac{1}{2}\log\left(2\pi e \left(2^{2R}+  N\right)\right)\nonumber\\&+\sum_{m=1}^M h(U|X_mQ)-h(U|X_\M )\nonumber\\
\leq &MC -(M\!-\!1) \frac{1}{2}\log\left(2\pi e \left(2^{2R}+  N\right)\right)\nonumber\\&+\sum_{m=1}^M h(U|X_m)-h(U|X_\M ).\label{eq:MI-bound}
\end{align}

Now consider the second term in   \eqref{thm:upperbound} with $\S=\varnothing$: 
\begin{align}
R\leq& I(X_\M;Y|Q) \nonumber\\
\leq& I(X_\M;Y)\nonumber\\
=& h(Y)-h(Y|X_\M) \label{eq:MACupp}
\end{align}
Note that the RHSs of \eqref{eq:MI-bound}, \eqref{eq:MACupp} are both concave in $p(x_\M)$ and symmetric with respect to $X_1,\ldots,X_M$. Therefore, a symmetric $p(x_\M)$ maximizes them. Let  $\mathbf{K}$ denote the covariance matrix of an  optimal symmetric solution. We have
\begin{align}
&\mathbf{K}_{ii}=P ,\quad \mathbf{K}_{ij}=P\rho.
\end{align}
Using the conditional version of the maximum entropy lemma  \cite{Thomas87}, we can  upper bound the differential entropies that appear with positive sign in \eqref{eq:MI-bound} and \eqref{eq:MACupp} by their Gaussian counterparts, and $h(U|X_\M)$ and $h(Y|X_\M)$ can be written explicitly because the  channels from $X_\M$ to $U$ and $Y$ are Gaussian. We  obtain \eqref{bound:R} and \eqref{bound:cutR}. Note that the RHSs of both bounds are increasing in $P$ and therefore there is no loss of generality in choosing $\mathbf{K}_{ii}=P$ (among $\mathbf{K}_{ii}\leq P$).
\end{proof}

The upper bound of Theorem \ref{cor-up-Gaussian} and the  lower bound of Theorem \ref{cor:lb-Gaussian} are plotted in Fig.~\ref{fig:Gaussian8a} for $M=3$ and $P=1$, and they are compared with the lower bounds of message sharing \cite{DaiYu14}, and  compression \cite{Parkatal13}. One sees that our lower and upper bounds are close and they match over a wide range of $C$. Moreover, establishing partial cooperation among the relays through Marton's coding offers significant gains. Fig.~\ref{fig:Gaussian8b} plots the capacity bounds for $P=1$ and different values of $M$.

      \begin{figure}[ht!]
    \centering
       \input{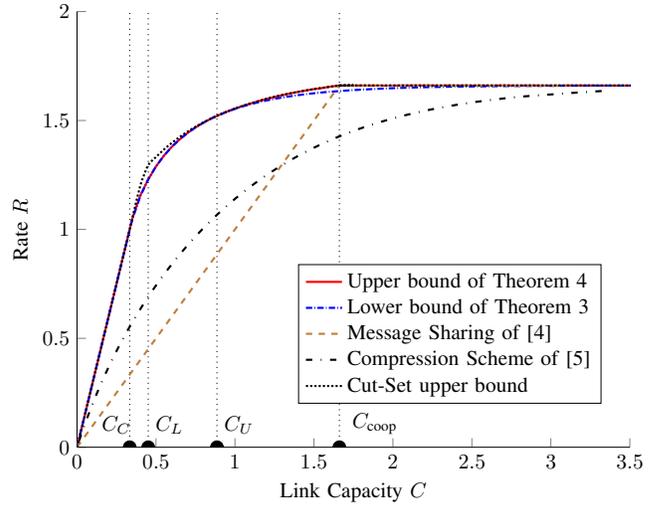}
        \caption{Capacity bounds as functions of $C$ ($M=3$, $P=1$).}
\label{fig:Gaussian8a}
    \end{figure}

\begin{figure}[ht!]
    \centering
       \input{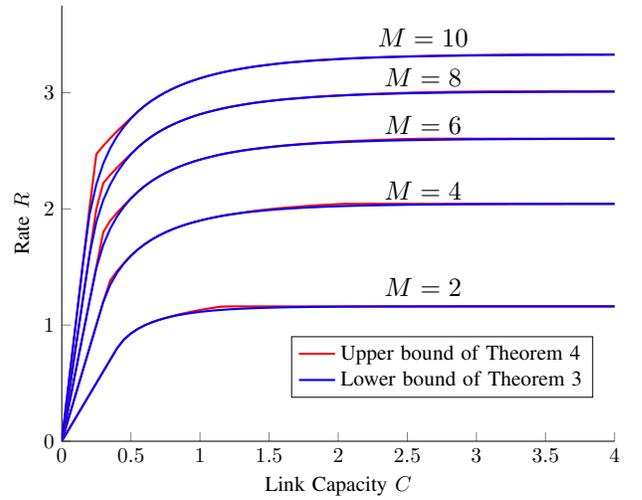}
    \caption{Capacity bounds as functions of $C$ for $P=1$ and different values of $M$.}
        \label{fig:Gaussian8b}
\end{figure}
\allowdisplaybreaks
We next compare the lower and upper bounds. Let
\begin{align}
&C_C=\frac{1}{2M}\log\left(1+PM\right)\label{eq:CC}\\
&C_L=\frac{1}{2M}\log\left(\frac{1+\frac{M^2}{2}P}{\left(\frac{M}{2(M-1)}\right)^{M-1}\frac{M}{2}}\right)\\
&C_U=\frac{1}{2M}\log\left(\frac{1+MP\left(1+(M-1)\rho^{(2)}\right)}{\left(1-\rho^{(2)}\right)^{M-1}\left(1+(M-1)\rho^{(2)}\right)}\right)\label{eq:CU}
\end{align}
where
\begin{align}
&\rho^{(2)}=\frac{M-2-\frac{1}{P}+\sqrt{(M-2-\frac{1}{P})^2+4(M-1)}}{2(M-1)}.
\end{align}
\begin{theorem}
\label{prop:range}
The lower bound of Theorem \ref{cor:lb-Gaussian} matches the upper bound of Theorem \ref{cor-up-Gaussian} if
\begin{align}
C\leq C_C\label{eq:cutset0}
\end{align}
or if
\begin{align}
C_L\leq C\leq C_U\label{eq:range}
\end{align}
where $C_C,C_L,C_U$ are defined in \eqref{eq:CC}--\eqref{eq:CU}.
\end{theorem}
\begin{remark}
Theorem \ref{prop:range} recovers \cite[Theorem~5]{SaeediKramer16} for $M=2$.
\end{remark}
\begin{remark}
For $C\leq C_C$, no cooperation is needed among the transmitters and the capacity is equal to~{$MC$}.
\end{remark}
\begin{remark}
For $C$ large enough, full cooperation is possible through superposition coding and the capacity is 
\begin{align}
C^{(M)}=\frac{1}{2}\log\left(1+M^2P\right), \qquad C\geq C_{\text{coop}}\label{eq:Cfull}
\end{align}
where
\begin{align}
C_{\text{coop}}=\frac{1}{2}\log\left(1+M^2P\right).
\end{align}
The rate \eqref{eq:Cfull} is not achievable by Theorem \ref{cor:lb-Gaussian} except when $C\to \infty$. This rate is achievable by message sharing.
\end{remark}
\begin{proof}[Proof of Theorem \ref{prop:range}]
To find regimes of $P$ and $C$ for which the lower and upper bounds match, we mimic the analysis in \cite[Appendix F]{SaeediKramer16}. Consider the lower bound in Theorem~\ref{cor:lb-Gaussian}, and in particular its maximum achievable rate $R^{(\ell)}$ which is attained by $\rho^{(\ell)}$, see \eqref{eq:rl0}--\eqref{def:Rl}. If \eqref{eq:cutset0} holds, we have $\rho^{(\ell)}=0$, $R^{(\ell)}=MC$, and thus the cut-set bound is achieved. Otherwise, we proceed as follows.

Consider \eqref{bound:R}. Since $R\geq R^{(\ell)}$, and using the definition of $R^{(\ell)}$ in \eqref{def:Rl}, we can further upper bound \eqref{bound:R} and obtain
\begin{align}
R\leq  &MC -(M\!-\!1) \frac{1}{2}\log\left( 1+N+PM(1+(M-1)\rho^{(\ell)}\right)\nonumber\\\quad&+ \frac{M}{2}\log\left(1\!+\!N\!+\!P((M\!-\!1)(1\!-\!\rho)(1\!+\!(M\!-\!1)\rho))\right)\nonumber\\&-\frac{1}{2}\log(1+N).
\label{bound:notitght}
\end{align}
Call the RHS of \eqref{bound:notitght} $g_1(\rho)$ and the RHS of \eqref{bound:cutR} $g_2(\rho)$. Fix $N$ as a function of $\rho^{(\ell)}$ such that 
\begin{align}
I^{\left(G,\rho^{(\ell)}\right)}(X_\M|U)=0\label{eq:bigif}
\end{align} 
where  $I^{\left(G,\rho^{(\ell)}\right)}(X_\M|U)$ is $I(X_\M|U)$ evaluated for a fully symmetric Gaussian distribution with correlation factor $\rho^{(\ell)}$. One can verify that the following choice of $N$  satisfies \eqref{eq:bigif}:
\begin{align}
N=P\frac{(1-\rho^{(\ell)})(1+(M-1)\rho^{(\ell)})}{\rho^{(\ell)}}-1.
\end{align}
The right inequality in \eqref{eq:range} ensures $N\geq 0$.

With this choice of $N$,  $g_1(\rho)$ is exactly equal to 
\[MC-I^{\left(G,\rho^{(\ell)}\right)}(X_\M)\] 
 at $\rho=\rho^{(\ell)}$. Note that $\rho^{(\ell)}$ is  defined in \eqref{def:rl}, and thus $g_1(\rho)$ crosses $g_2(\rho)$ at $\rho^{(\ell)}$. Since $g_2(\rho)$ is increasing in $\rho$,  the maximum admissible rate by \eqref{bound:cutR} and \eqref{bound:notitght} matches $R^{(\ell)}$ if  $g_1(\rho)$ is non-increasing for $\rho\geq \rho^{(\ell)}$. This is ensured by the left inequality in \eqref{eq:range}.
\end{proof}
\section{Concluding Remarks}
We studied the downlink of symmetric C-RANs with multiple relays and a single receiver, and established lower and upper bounds on the capacity. The lower bound uses Marton's coding  to establish partial cooperation among the relays and improves on schemes that are based on message sharing and compression for scalar Gaussian C-RANs  (see Fig. \ref{fig:Gaussian8a}). 
The upper bound generalizes the bounding techniques of \cite{KangLiuChong15,SaeediKramer16}. When specialized to symmetric Gaussian C-RANs, the lower and upper bounds meet over a wide range of $C$ and this range gets large as $M$ and/or $P$ get large.

\ifCLASSOPTIONcaptionsoff
  \newpage
\fi

\appendices
\section{Analysis of the Achievable Scheme}

\label{chert}
The  scheme fails only if one of the following events occur:
\begin{itemize}
\item \eqref{en:jt1} does not hold for any index tuple $(w_1^{\prime\star},\ldots, w_M^{\prime\star})$. This event has a vanishing error probability as $n\to \infty$ \cite[Lemma 8.2]{ElGamalKim} if we have \eqref{en:en}.
\item \eqref{de:jt2} does not hold for the original indices $(w_1,w^{\prime\star}_1,\ldots,w_M,w^{\prime\star}_M)$. This event has a vanishing error probability as $n\to\infty$ by \eqref{en:jt1} and the Law of Large Numbers.
\item \eqref{de:jt2} holds for indices $(\tilde{w}_1,\tilde{w}^{\prime}_1,\ldots,\tilde{w}_M,\tilde{w}^{\prime}_M)$ where $(\tilde{w}_1,\ldots,\tilde{w}_M)\neq(w_1,\ldots,w_M)$. We show that this event has a vanishing error probability as $n\to\infty$ if we have \eqref{en:de}.  Since the codebook is symmetric with respect to all messages, we assume without loss of generality that $w_1=\ldots=w_M=1$ and $w_1^\prime=\ldots=w_M^\prime=1$. 
Fix the sets $\S_1,\S_2,\S_3 \subset\M$ such that $\S_1\cup\S_2\cup\S_3=\M$.
Consider the case where 
\begin{align}
\label{eq:Cs}
\begin{array}{ll}
\tilde{w}_i=1, \tilde{w}_i^\prime=1& \forall i\in\S_1\\
\tilde{w}_i=1, \tilde{w}_i^\prime\neq1& \forall i\in\S_2\\
\tilde{w}_i\neq 1& \forall i\in\S_3.
\end{array}
\end{align}
We denote  \eqref{eq:Cs} by  $W(\S_1^3)$. We have
\begin{align}
&\Pr\!\left(\bigcup_{\substack{W(\S_1^3)}}\!\! (X_1^n(\tilde{w}_1,\tilde{w}^{\prime}_1),\ldots,X_M^n(\tilde{w}_M,\tilde{w}^{\prime}_M),Y^n)\in \mathcal{T}_\epsilon^n \!\right)\nonumber\\
&\leq 2^{n\sum_{m\in\S_3}(R_m+R^\prime_m)+n\sum_{m\in\S_2}R^\prime_m}\nonumber\\&\quad\quad \times\!\!\!\!\sum_{(x_1^n,\ldots,x_M^n,y^n)\in \mathcal{T}_\epsilon^n}\!\!\!\!p_{X^n_{\S_1}\tilde{X}^n_{\S_2}\hat{X}^n_{\S_3}Y^n}(x_{\S_1}^n,x_{\S_2}^n,x_{\S_3}^n,y^n)\nonumber\\
&\stackrel{(a)}{\leq} 2^{n\sum_{m\in\S_3}(R_m+R^\prime_m)+n\sum_{m\in\S_2}R^\prime_m}\nonumber\\&\quad\times\! 2^{n(H(X_\M Y)-H(X_{\S_1}Y)-\sum_{m\in\S_2\cup\S_3}\!H(X_{m})-\delta(\epsilon))}\!
\end{align}
where  $\tilde{X}_{\S_2}^n$ denotes $\{X_m^n(\tilde{w}_m,\tilde{w}_m^\prime)\}_{m\in\S_2}$ and $\hat{X}_{\S_3}^n$ denotes $\{X_m^n(\tilde{w}_m,\tilde{w}_m^\prime)\}_{m\in\S_3}$.
In step $(a)$, we use that (i) $(X_{\S_1}^nY^n)$ is ``almost independent" of $(\tilde{X}_{\S_2}^n,\hat{X}_{\S_3}^n)$ and (ii) the random sequences $X_m^n(\tilde{w}_m,\tilde{w}_m^\prime),\ {m\in\S_2}$, and $X_m^n(\tilde{w}_m,\tilde{w}_m^\prime), \ m\in\S_3$, are mutually ``almost independent". Note that we use the term almost independent, rather than independent, because we have assumed $w_1=\ldots=w_M=1$ and $w_1^\prime=\ldots=w_M^\prime=1$; i.e., we implicitly have a conditional probability and conditioned on $w_1^\prime=\ldots=w_M^\prime=1$, (i)-(iii) may not hold if we insist on exact independence. 
This has been dealt with in \cite{MineroLimKim15,GroverWagnerSahai10,SaeediPrabhakaran14}. Following similar arguments, one can show that (i) and (iii) hold with ``almost independence". 
The probability of the considered error event is thus arbitrarily small for large enough $n$ if
\begin{align}
&\sum_{m\in\S_3}(R_m+R^\prime_m)+\sum_{m\in\S_2}R^\prime_m\nonumber\\&\quad \quad\quad<I(X_\M;Y|X_{\S_1})+I(X_\M)-I(X_{\S_1}).\label{eq:achS1}
\end{align}
Inequality \eqref{eq:achS1} is satisfied by \eqref{en:de} when we choose $\bar{S}=\S_1$. Note that the inequalities with $\S_2\neq \varnothing$ are redundant.
\end{itemize}

%
%
\section{Simplification for  Symmetric Networks with Symmetric Distributions}
\label{ap-simplify}
Choose $R_m=\tilde{R}$ and $R_m^\prime=\tilde{R}^\prime$  for all $1,\ldots,M$. The problem defined by \eqref{en:set1}, \eqref{en:set2}, \eqref{en:en}, \eqref{en:de} simplifies  using the definition in \eqref{symm1}:
\begin{align}
&\tilde{R},\tilde{R}^\prime\geq 0\\
&\tilde{R}+\tilde{R}^\prime\leq C\\
&M\tilde{R}=R\\
&|\S|\tilde{R}^\prime\geq I(X_\S)\qquad \forall \S\subseteq \M\label{torelax1}\\
&|\S|(\tilde{R}+\tilde{R}^\prime)\leq I(X_\S;Y|X_{\bar\S})\!+\!I(X_\M)\!-\!I(X_{\bar{\S}})\quad  \forall\S\subseteq\M.\label{torelax2}
\end{align}
We prove that the tightest inequality in \eqref{torelax1} and \eqref{torelax2} is given by $\S=\M$. Eliminating $\tilde{R}^\prime$ from the remaining inequalities concludes the proof.

Let $\S$ be any subset of $\M$ and $s_0$ be an element of $\M$ such that $s_0\notin\S$. This is possible if $\S\neq\M$. Define $\T=\M\backslash\{\S\cup s_0\}$. 
We show 
\begin{align}
\frac{I(X_\S)}{|\S|}\leq \frac{I(X_{S\cup\{s_0\}})}{|\S|+1}\label{simplify}
\end{align}
and  
\begin{align}
\frac{1}{|\S|+1}&\left(I(X_{\S\cup\{s_0\}};Y|X_{\T})+I(X_\M)-I(X_{\T})\right)\nonumber\\&\leq \frac{1}{|\S|}\left(I(X_\S;Y|X_{\T}X_{s_0})+I(X_\M)-I(X_{\T\cup s_0})\right).
\end{align}

The following  equalities come in handy in the proof:
\begin{align}
&I(X_{\S\cup \T})=I(X_\S)+I(X_\T)+I(X_\S;X_\T)\label{eq:useful1}\\
&I(X_\S)=\sum_{j=1}^{|\S|-1}I(X_{s_1}\ldots X_{s_j};X_{s_{j+1}}).\label{eq:useful2}
\end{align}
Suppose $\S=\{s_1,\ldots,s_{|\S|}\}$ and $\S^\prime=\{s_0,s_1,\ldots,s_{|\S|-1}\}$. We have
\begin{align}
|\S|I(X_{\S\cup s_0})&\!=\!|\S|I(X_\S)+|\S|I(X_\S;X_{s_0})\nonumber\\
&\!\geq\! |\S|I(X_\S)+\sum_{j=1}^{|\S|}I(X_{s_1}\ldots X_{s_j};X_{s_0})\nonumber\\
&\!\stackrel{(a)}=\!  |\S|I(X_\S)+\sum_{j=0}^{|\S|-1}I(X_{s_0}\ldots X_{s_j};X_{s_{j+1}})\nonumber\\
&\!\stackrel{(b)}{=}\!|\S|I(X_\S)\!+\!I(X_{\S^\prime})\!+\!I(X_{s_0}\!\ldots X_{s_{|\S|\!-\!1}};X_{s_{|\S|}}\hspace{-.05cm})\nonumber\\
&\!\geq\! |\S|I(X_\S)+I(X_{\S^\prime})\nonumber\\
&\!\stackrel{(c)}{=}\!(|\S|+1)I(X_\S)
\end{align}
where $(a)$ and $(c)$ hold by \eqref{symm1} and $(b)$ is by \eqref{eq:useful2}, written for~$\S^\prime$. 

Similarly, we have
\begin{align}
(|\S|\!+\!1)&\left(I(X_\S;Y|X_{\T}X_{s_0})+I(X_\M)-I(X_{\T\cup s_0})\right)\nonumber\\&\hspace{-1.15cm}-|\S|\left(I(X_{\S\cup\{s_0\}};Y|X_{\T})+I(X_\M)-I(X_\T)\right)\nonumber\\
&=I(X_\S;Y|X_\T X_{s_0})-|\S|I(X_{s_0};Y|X_\T)\nonumber\\&\quad+I(X_\M)\!-\!(|\S|\!+\!1)I(X_{\T\cup s_0})\!+\!|\S|I(X_\T)\nonumber\\
&\stackrel{(a)}{=}I(X_\S;Y|X_\T X_{s_0})-|\S|I(X_{s_0};Y|X_\T)\nonumber\\&\quad+I(X_\S)+I(X_{\T\cup s_0})+I(X_\S;X_{\{\T\cup s_0\}})\nonumber\\&\quad-(|\S|\!+\!1)I(X_{\T\cup s_0})+|\S|I(X_\T)\nonumber\\
&\stackrel{(b)}{=}I(X_\S;Y|X_\T X_{s_0})-|\S|I(X_{s_0};Y|X_\T)\nonumber\\&\quad+I(X_\S)\!+\!I(X_\S;X_{\T\cup s_0})\!-\!|\S|I(X_\T;X_{s_0})\nonumber\\
&=I(X_\S;YX_\T X_{s_0})\!-\!|\S|I(X_{s_0};YX_\T)\!+\!I(X_\S)\nonumber\\
&=H(X_\S)+I(X_\S)-H(X_\S|YX_\T X_{s_0})\nonumber\\&\quad-|\S|\left(H(X_{s_0})-H(X_{s_0}|YX_\T)\right)\nonumber\\
&\stackrel{(c)}{=} |\S|H(X_{s_0}|YX_\T)-H(X_\S|YX_\T X_{s_0})\nonumber\\
&\stackrel{(d)}\geq 0
\end{align}
where $(a)$ and $(b)$ are by \eqref{eq:useful1}, $(c)$ follows from \eqref{symm1} and $(d)$ follows from \eqref{symm1} and the symmetry of the channel in \eqref{chansym}. 
\section{Proof of Theorem \ref{thm:upperbound-statement}}
\label{ap-thm:upperbound-statement}
We first prove the multi-letter bound in \eqref{eq:ubound1} using Fano's inequality and the data processing inequality. For any $\epsilon>0$ we have
\begin{align}
nR\stackrel{(a)}{\leq}& H(\mathbb{V}_1,\ldots,\mathbb{V}_M)+n\epsilon\nonumber\\
=& \sum_{m=1}^MH(\mathbb{V}_m)-I(\mathbb{V}_\M)+n\epsilon\nonumber\\
\leq& \sum_{m=1}^MnC-I(\mathbb{V}_\M)+n\epsilon\nonumber\\
\stackrel{(b)}{\leq}& \sum_{m=1}^MnC-I(X^n_\M)+n\epsilon\label{eq:multi1}
\end{align}
where $(a)$ is by Fano's inequality and $(b)$ is by the data processing inequality as follows:
\begin{align}
I(\mathbb{V}_\M)=&\sum_{m=1}^nI(\mathbb{V}_1\ldots \mathbb{V}_m;\mathbb{V}_{m+1})\nonumber\\
\leq &\sum_{m=1}^nI(X^n_1\ldots X^n_m;X^n_{m+1})\nonumber\\
=&I(X^n_\M).
\end{align}

Similarly, for any subset $\S\subset\M$ and $\epsilon>0$, we have
\begin{align}
nR\leq& I(W;Y^n)+n\epsilon\nonumber\\
\leq& I(W;Y^n\mathbb{V}_{\S}X^n_\S)+n\epsilon\nonumber\\
=& I(W;\mathbb{V}_{\S}X^n_\S)+ I(W;Y^n|\mathbb{V}_{\S}X^n_\S)+n\epsilon\nonumber\\
\stackrel{(a)}{=}& I(W;\mathbb{V}_{\S})+ I(W;Y^n|\mathbb{V}_{\S}X^n_\S)+n\epsilon\nonumber\\
\leq& H(\mathbb{V}_\S)+ I(WX_\M^n;Y^n|\mathbb{V}_{\S}X^n_\S)+n\epsilon\nonumber\\
\stackrel{(b)}{\leq}& n|\S|C+ I(X_\M^n;Y^n|X^n_\S)+n\epsilon
\end{align}
where $(a)$ is because $W-\mathbb{V}_\S-X^n_\S$ forms a Markov chain and $(b)$ is because conditioning does not increase entropy and because $W\mathbb{V}_\M-X^n_\M-Y$ forms a Markov chain.

Next, we upper bound \eqref{reg:3user1} using \eqref{rewriteU} and single-letterize \eqref{rewriteU} as follows:
\begin{align}
nR&\leq nMC -(M\!-\!1)\sum_{i=1}^n H(U_i|U^{i-1})\nonumber\\&\quad+\sum_{m=1}^M \sum_{i=1}^nH(U_i|X_{\M}^nU^{i-1})-\sum_{i=1}^nH(U_i|X_{\M}^nU^{i-1})\nonumber\\
&\leq nMC -(M\!-\!1)\sum_{i=1}^n H(U_i|U^{i-1})\nonumber\\&\quad+\sum_{m=1}^M\! \sum_{i=1}^n\!H(U_i|X_{m,i}U^{i-1})-\sum_{i=1}^nH(U_i|X_{\M,i}U^{i-1})\nonumber\\
&=nMC -n(M\!-\!1) H(U_T|U^{T-1}T)\nonumber\\&\quad+\!n\!\!\sum_{m=1}^M \!H(U_T|X_{m,T}U^{T\!-\!1}T)\!-\!nH(U_T|X_{\M,T}U^{T\!-\!1}T)\nonumber\\
&=n\!\left(\!\!MC\! -\!(M\!-\!1) H(U|Q)\!\!\phantom{\sum_{m=1}^M\! \!H(U|X_mQ)\!-\!H(U|X_\M Q)}\right.\nonumber\\&\qquad\qquad\qquad \left.+\sum_{m=1}^M\! \!H(U|X_mQ)\!-\!H(U|X_\M Q) \!\!\right)
\end{align}
where $T$ is a uniform random variable on the set $\{1,\ldots,n\}$ independent from everything in the system, $Q$ is defined by $(U^{T-1}T)$ and $X_1,\ldots,X_M$, $Y$, and $U$ are defined by $X_{1,T},\ldots,X_{M,T}$, $Y_T$, and $U_T$, respectively. Note that $U-Y-X_1X_2-Q$ forms a Markov chain.

Finally, we expand \eqref{eq:ubound2} as follows.
\begin{align}
nR&\leq n|\S|C+I(X_{{\M}}^n;Y^n|X_\S^n)\nonumber\\
&= n|\S|C+\sum_{i=1}^nH(Y_i|X_{\S}^nY^{i-1})-\sum_{i=1}^nH(Y_i|X_{\M}^nY^{i-1})\nonumber\\
&\stackrel{(a)}{=} n|\S|C+\sum_{i=1}^nH(Y_i|X_{\S}^nY^{i-1}U^{i-1})\nonumber\\&\quad-\sum_{i=1}^nH(Y_i|X_{\M}^nY^{i-1}U^{i-1})\nonumber\\
&\stackrel{(b)}{\leq} n|\S|C+\sum_{i=1}^nH(Y_i|X_{\S,i}U^{i-1})-\sum_{i=1}^nH(Y_i|X_{\M,i}U^{i-1})\nonumber\\
&\stackrel{(c)}{=} n\left(|\S|C+I(X_{\bar{\S}};Y|X_\S Q)\right)
\end{align}
where $(a)$ is because $Y_i-X_\S^nY^{i-1}-U^{i-1}$ forms a Markov chain, $(b)$ is because $Y_i-X_{\M,i}-U^{i-1}$ forms a Markov chain, and $(c)$ is by a standard time sharing argument.
The cardinality of $\mathcal{Q}$ is bounded using the Fenchel-Eggleston-Carath\'eodory theorem \cite[Appendix A]{ElGamalKim} (see also \cite[Appendix B]{cribbing85}). 

\bibliographystyle{IEEEtran}
\bibliography{bibographyMACNRelay}

\end{document}